\documentclass[11pt,a4paper]{amsart}

\usepackage[margin=1in]{geometry}
\usepackage[T1]{fontenc}
\usepackage{amsmath,amssymb,amsthm,mathtools}
\usepackage{booktabs}
\usepackage[hidelinks]{hyperref}

\newtheorem{theorem}{Theorem}[section]
\newtheorem{lemma}[theorem]{Lemma}

\newtheorem{proposition}[theorem]{Proposition}
\theoremstyle{definition}
\newtheorem{definition}[theorem]{Definition}
\theoremstyle{remark}
\newtheorem{remark}[theorem]{Remark}

\newcommand{\E}{\mathbb{E}}
\newcommand{\Prob}{\mathbb{P}}
\newcommand{\Czero}{\mathcal C}
\newcommand{\Cs}{C_{\Sigma}}
\newcommand{\Csy}{C_{\mathrm{sym}}}
\newcommand{\one}{\mathbf{1}}
\newcommand{\cE}{\mathcal{E}}
\DeclareMathOperator{\dist}{d_H}

\title{A Continuous Projection Converse and a Weighted Construction 
for Uniquely Decodable Code Pairs}
\author{Lei Yu}
\thanks{L. Yu is with the School of Statistics and Data Science, LPMC, KLMDASR,
	and LEBPS, Nankai University, Tianjin 300071, China (e-mail: leiyu@nankai.edu.cn).
}

\begin{document}

\subjclass[2020]{94A40, 94B65, 05D05}
\keywords{Binary adder channel, zero-error capacity, uniquely decodable code pairs,
projection bounds, weighted code constructions}

\begin{abstract}
We study the maximum sum rate of uniquely decodable code pairs.  
A weighted refinement of a complement-gluing
construction yields an explicit code of length 672 and sum rate exceeding
$1.318639029203$. An iterated projection argument, combined with a
conditional entropy estimate and a justified continuous limit, gives an upper bound of $1.480063425539$. 
\end{abstract}

\maketitle

\section{Introduction}\label{sec:introduction}

The binary adder channel is a basic model of noiseless multiple-access
communication: two senders independently choose binary input words, and
the receiver observes their coordinatewise integer sum. Although the
channel itself introduces no randomness, different input pairs may
produce the same output. Zero-error communication therefore asks for two
codebooks whose pairwise sums are all distinct. Such codebooks are also
called uniquely decodable code pairs, and their study connects coding
theory with extremal set theory and additive combinatorics.

\subsection{Background and related work}

The distinction between vanishing error and zero error is substantial.
The usual Shannon capacity region is the pentagon\footnote{All information rates are measured by bits, and the base of  ``log''  in this paper is $2$.}
$0\le R_i\le1$, $R_1+R_2\le3/2$; independent fair inputs attain its
sum-rate constraint; see, for example, \cite{OS}. Requiring injectivity for every message pair leads
to a different problem. Lindstr\"om's early work~\cite{Lindstrom}
established basic constructions, including the sum rate
$\tfrac12\log 6$. Mattas and \"Osterg\aa rd~\cite{MO} subsequently
obtained a six-coordinate code with message product 240.
Balogh, Nguyen, \"Osterg\aa rd, White, and Wigal~\cite{BNOWW}
used complement gluing to obtain the rate $1.318446971988\ldots$.
The lower bound in this paper refines that construction by replacing
Hamming weight with a positive integer weighted score.

Converse arguments must use the Cartesian-product structure of the
message set, since counting ternary outputs alone only gives $\log3$.
Entropy bounds, restrictions to coordinates, and bounds on correlated
binary rectangles provide complementary ways of exploiting this
structure. Ordentlich and Shayevitz~\cite{OS} developed an outer bound
using multiset-union-free families and conditional entropy. Austrin,
Kaski, Koivisto, and Nederlof~\cite{AKKN} obtained a sharper explicit
bound for highly unbalanced code pairs. The rectangle estimates of
Ordentlich, Polyanskiy, and Shayevitz~\cite{OPS} and the subsequent
results of Yu, Anantharam, and Chen~\cite{YAC} further strengthen the
converse near a unit individual rate.

Zhang~\cite{Zhang} recently broke the classic $3/2$ bound  for sum-capacity. 
Precisely, he obtained an explicit sum-capacity upper
bound of approximately $1.4884$ by projecting a unique-sum pair to a
lower-dimensional unique-sum system and applying a relaxation of the
Ordentlich--Shayevitz condition. This provides a direct antecedent to
the projection approach developed here. Our converse iterates the
section-extraction step and passes to a continuous trajectory. An
important distinction in both approaches is that the projected object
is a family of code pairs with disjoint output sets, rather than a
single Cartesian-product code pair.

The corner regime also has its own constructions. Kasami, Lin, Wei,
and Yamamura~\cite{KLWY} obtained asymptotic rates $(1,1/4)$;
Wiman~\cite{Wiman} improved the second coordinate to $0.2563$.
These boundary constructions address a different part of the region
from the nearly balanced construction considered here.

\subsection{Main results and proof strategy}

Let $\Cs$ denote the supremum of achievable zero-error sum rates and
let $\Csy$ denote the supremum of rates achievable by both users
simultaneously. Precise definitions appear in Section~\ref{sec:model}.

Define  
\begin{align}
	J(q)&:=2h\!\left(\frac{1-\sqrt{1-2q}}{2}\right)-q,
\; 0\le q\le\frac12.\label{eq:J}
\end{align}
Here, $h$ is the binary entropy function. 

By extending   the projection  approach in \cite{Zhang} to a continuous version, we establish the following converse bound. 
\begin{theorem}[Continuous projection converse]\label{thm:continuous}
	Let $0<a<b<1$ with $a<1/2$, and let
	$q\in C^2([a,b])$ take values in $(0,1/2)$. Define
	\[
	U(s)=s+(1-s)J(q(s)).
	\]
	Suppose
	\begin{equation}\label{eq:potential-conditions}
		U(b)\ge(1-b)\log3,\qquad U'(s)\le\log q(s)
		\quad(a\le s\le b).
	\end{equation}
	Then
	\begin{equation}\label{eq:continuous-bound}
		\Cs\le h(a)-\frac12h(2a)+a+(1-a)J(q(a)).
	\end{equation}
\end{theorem}

By choosing proper $(a,b,q)$, we obtain the following numerical upper bound. We also obtain an improved numberical lower bound by constructing an explicit code. 

\begin{theorem}[Main result]\label{thm:main}
For the two-user binary adder channel without feedback,
\begin{equation}\label{eq:main-summary}
1.318639029203<\Cs<1.480063425539.
\end{equation}
Moreover, $\Csy=\Cs/2$, and hence
\[
0.6593195146015<\Csy<0.7400317127695.
\]
The lower bound is attained strictly by a specified length-672 code.
The upper bound applies to all blocklengths through the supremum
characterization of sum capacity.
\end{theorem}

The bounds  on the sum-capacity improve the  currently best upper bound    $1.4884$ due to Zhang  \cite{Zhang} and  the currently best lower bound  $1.318446971988\ldots$ due to Balogh, Nguyen, \"Osterg\aa rd, White, and Wigal~\cite{BNOWW}.  

The achievability argument retains two parts of a tensor-power codebook,
complements one part, and separates their output sets by an additive
weighted score. Exact generating functions count the retained words.
The converse follows a different route. A section with small mean
input disagreement has a strong conditional entropy bound; a section
with large disagreement contains a sufficiently large further section.
Iterating this alternative gives a finite recursion. Passing first to
large blocklength and then to a fine projection grid produces a
differential inequality that admits an explicit trajectory.

Sections~\ref{sec:entropy} and~\ref{sec:projection} establish the two
local estimates. Sections~\ref{sec:finite}--\ref{sec:explicit} prove
and evaluate the converse. Section~\ref{sec:achievability} gives the
construction. 

\section{Model, notation, and elementary bounds}\label{sec:model}

The channel is
\[
Z=X_1+X_2,\qquad X_1,X_2\in\{0,1\},\qquad Z\in\{0,1,2\},
\]
where addition is over the integers. All logarithms and entropies are to
base $2$, unless $\ln$ is written explicitly. The encoders do not have
feedback.

A length-$n$ zero-error code is a pair of nonempty sets
$A,B\subseteq\{0,1\}^n$ such that $(\mathbf{a},\mathbf{b})\mapsto \mathbf{a}+\mathbf{b}$ is injective on
$A\times B$. Equivalently,
\begin{equation}\label{eq:diff}
(A-A)\cap(B-B)=\{\mathbf{0}\}.
\end{equation}
Indeed, a collision $\mathbf{a}+\mathbf{b}=\mathbf{a}'+\mathbf{b'}$ is equivalent to
$\mathbf{a}-\mathbf{a}'=\mathbf{b}'-\mathbf{b}$; difference sets are symmetric about the origin.

\begin{definition}[Zero-error capacity region]
For a length-$n$ code pair $(A,B)$, write
$r_1=n^{-1}\log|A|$ and $r_2=n^{-1}\log|B|$. Define
\begin{equation}\label{eq:region-definition}
\Czero:=\overline{\bigcup_{n\ge1}\ \bigcup_{(A,B)\text{ zero-error}}
[0,r_1]\times[0,r_2]}.
\end{equation}
The sum and symmetric capacities are, respectively,
\[
\Cs:=\max_{(R_1,R_2)\in\Czero}(R_1+R_2),\qquad
\Csy:=\max\{R:(R,R)\in\Czero\}.
\]
\end{definition}

\begin{proposition}[Basic properties]\label{prop:region}
The region $\Czero$ is compact, convex, downward closed, and invariant
under exchanging its coordinates. It contains $(1,0)$ and $(0,1)$
and is a subset of $[0,1]^2$.
\end{proposition}
\begin{proof}
The individual bounds follow from $|A|,|B|\le2^n$; compactness and
downward closure follow from the definition. Exchanging codebooks
exchanges rates. The pair $(\{0,1\}^n,\{0^n\})$ gives $(1,0)$.
Concatenating codes of lengths $n$ and $m$ gives the length-weighted
average of their rate pairs. Repeating the two codes makes these
weights dense in $[0,1]$, so taking the closure proves convexity.
\end{proof}

Define
\[
M(n):=\max\bigl\{|A||B|: A,B\subseteq\{0,1\}^n
\text{ form a zero-error code}\bigr\}.
\]
Concatenation gives $M(n+m)\ge M(n)M(m)$, while $M(n)\le 3^n$.
Consequently, the zero-error sum capacity satisfies
\begin{equation}\label{eq:capacity}
\Cs=\lim_{n\to\infty}\frac{1}{n}\log M(n)
=\sup_{n\ge1}\frac{1}{n}\log M(n).
\end{equation}

\begin{proposition}[Classical elementary bounds]\label{prop:classical}
One has
\[
\frac12\log6\le\Cs\le\frac32.
\]
\end{proposition}
\begin{proof}
The length-two pair $A=\{00,11\}$, $B=\{00,01,10\}$ has six
distinct sums, proving the lower bound. For the upper bound, take
independent uniform $X\in A$, $Y\in B$. Injectivity and entropy
subadditivity give
\[
\log(|A||B|)=H(X+Y)\le\sum_{i=1}^n H(X_i+Y_i)\le\frac32n.
\]
The last inequality follows from Lemma~\ref{lem:bits} below:
$J\le J(1/2)=3/2$, and $L$ is decreasing on $[1/2,1]$.
Now use \eqref{eq:capacity}.
\end{proof}

\subsection{Sections of a code pair}
The converse studies pairs whose sum is fixed to one on a chosen
set of coordinates. Unlike a restriction of each codebook separately,
this operation can introduce dependence between the two inputs.

For $I\subseteq[n]:=\{1,\ldots,n\}$, define the section
\begin{equation}\label{eq:section}
\cE_I:=\{(a,b)\in A\times B:a_i+b_i=1\text{ for all }i\in I\},
\qquad N_I:=|\cE_I|.
\end{equation}
When $N_I>0$ and $|I|<n$, sample $(X,Y)$ uniformly from $\cE_I$ and set
\begin{equation}\label{eq:qI}
q_I:=\frac{1}{n-|I|}\E\dist(X_{I^c},Y_{I^c}).
\end{equation}
This conditioned distribution need not make $X$ and $Y$ independent.
The conditional independence needed below will be established explicitly.

\section{A conditional entropy estimate}\label{sec:entropy}

The following estimate records the dependence on the disagreement
probability. It is related to the conditional entropy approach in
\cite{OS}; the proof and its application to sections are included here.

Let
\[
h(p):=-p\log p-(1-p)\log(1-p),\qquad 0\le p\le1,
\]
with $0\log0:=0$. 
Recall 
\begin{align}
	J(q)&:=2h\!\left(\frac{1-\sqrt{1-2q}}{2}\right)-q.
\end{align}
Define
\begin{align}
L(q)&:=h(q)+1-q,\; 0\le q\le1,\label{eq:L}\\
F(q)&:=\begin{cases}
J(q),&0\le q\le\frac12,\\
L(q),&\frac12\le q\le1.
\end{cases}\label{eq:F}
\end{align}

\begin{lemma}[Entropy of the sum of independent bits]\label{lem:bits}
The function $F$ is concave on $[0,1]$, and $J$ is increasing on
$[0,1/2]$. For independent binary random variables $V,W$,
\begin{equation}\label{eq:bit-entropy}
H(V+W)\le F\bigl(\Prob(V\ne W)\bigr).
\end{equation}
\end{lemma}

\begin{proof}
Write $q=\Prob(V\ne W)$ and $p_i=\Prob(V+W=i)$ for $i\in\{0,2\}$.
Independence implies
\begin{equation}\label{eq:product}
p_0p_2
=\Prob(V=0,W=1)\Prob(V=1,W=0)
\le\frac{q^2}{4}.
\end{equation}
Also, $p_0+p_2=1-q$.

Suppose first that $q\le1/2$. For fixed $p_0+p_2$, the entropy of
$(p_0,q,p_2)$ increases as $p_0$ and $p_2$ become more balanced,
or equivalently as $p_0p_2$ increases. Under \eqref{eq:product}, its maximum
is therefore attained at $p_0p_2=q^2/4$. This distribution is realized by
the sum of two independent $\operatorname{Bernoulli}(t)$ variables $V',W’$, where
\[
t=\frac{1\pm \sqrt{1-2q}}{2}.
\]
By the identity $H(V',W’)=H(V'+W’) + H(V’,W'|V’+W')$,  the entropy of the sum is $H(V’+W')=2h(t)-q=J(q)$.
When $q\ge1/2$, maximizing entropy over all ternary distributions with
mass $q$ at one gives the upper bound $L(q)$.

By Proposition 3.3 of \cite{Wyner}, $J$
is increasing and concave. The function $L$ is concave, and
\[
J(1/2)=L(1/2)=\frac32,\qquad
J'(1/2-)=\frac{2}{\ln2}-1>-1=L'(1/2+).
\]
The derivative therefore jumps downward at the joining point, proving
concavity of $F$. 
\end{proof}

\begin{lemma}[Entropy bound for a section]\label{lem:section-entropy}
Let $I\subseteq[n]$ satisfy $N_I>0$, and write $k=|I|<n$ and $m=n-k$.
Then
\begin{equation}\label{eq:section-entropy}
\log N_I\le k+mF(q_I).
\end{equation}
In particular, for $0\le q_*\le1/2$,
\begin{equation}\label{eq:low-disagreement}
q_I\le q_*\quad\Longrightarrow\quad
\log N_I\le k+mJ(q_*).
\end{equation}
\end{lemma}

\begin{proof}
Sample $(X,Y)$ uniformly from $\cE_I$, and put
\[
U=X_I,\qquad Z=X_{I^c}+Y_{I^c}.
\]
Given $U=u$ with positive probability, $Y_I=\one-u$, and the conditional
support of $(X_{I^c},Y_{I^c})$ is the Cartesian product
\[
\{a_{I^c}:a\in A,\ a_I=u\}
\times
\{b_{I^c}:b\in B,\ b_I=\one-u\}.
\]
The conditional distribution is uniform on this product. Hence
\begin{equation}\label{eq:conditional-independence}
X_{I^c}\perp Y_{I^c}\mid U.
\end{equation}

The remaining output $Z$ determines $(X,Y)$: the omitted output
coordinates all equal one, and the full output uniquely determines the
input pair. Consequently,
\begin{equation}\label{eq:entropy-identities}
H(Z)=\log N_I,\qquad H(U\mid Z)=0.
\end{equation}
For $i\in I^c$, write $q_{i,u}=\Prob(X_i\ne Y_i\mid U=u)$.
Using \eqref{eq:conditional-independence}, Lemma~\ref{lem:bits}, and
concavity of $F$ yields
\begin{align*}
\log N_I
&=H(U)+H(Z\mid U)\\
&\le k+\sum_{i\in I^c}H(X_i+Y_i\mid U)\\
&\le k+\sum_{i\in I^c}\E F(q_{i,U})\\
&\le k+mF\!\left(\frac1m\sum_{i\in I^c}\E q_{i,U}\right)
=k+mF(q_I).
\end{align*}
If $q_I\le q_*\le1/2$, the last assertion follows because $F=J$ on
$[0,1/2]$ and $J$ is increasing there.
\end{proof}

\section{Extracting a further section}\label{sec:projection}
The next estimate extends the double-counting argument of
\cite[Lemma 3.5]{Zhang} to an arbitrary existing section and a general
mean-disagreement threshold. This form allows repeated extraction.

\begin{lemma}[Projection estimate]\label{lem:projection}
Suppose $N_I>0$, $m=n-|I|>0$, and $q_I\ge q>0$. For every integer
$\ell$ with $1\le\ell\le\lfloor qm\rfloor$, there exists
$K\subseteq I^c$ with $|K|=\ell$ such that
\begin{equation}\label{eq:projection}
N_{I\cup K}\ge
N_I\frac{\binom{\lfloor qm\rfloor}{\ell}}{\binom m\ell}.
\end{equation}
\end{lemma}

\begin{proof}
	Under the uniform distribution on $\cE_I$, let
	$W=\dist(X_{I^c},Y_{I^c})$. For each fixed pair
	$(a,b)\in\cE_I$, define its remaining disagreement set by
	\[
	D_I(a,b):=\{j\in I^c:a_j\ne b_j\}.
	\]
	Since the coordinates are binary, $a_j\ne b_j$ is equivalent to
	$a_j+b_j=1$. Consequently, for $K\subseteq I^c$,
	\[
	(a,b)\in\cE_{I\cup K}
	\quad\Longleftrightarrow\quad K\subseteq D_I(a,b).
	\]
	We now count the incidences $((a,b),K)$ with $(a,b)\in\cE_I$,
	$K\subseteq D_I(a,b)$, and $|K|=\ell$ in two ways. For each fixed
	$K$, there are $N_{I\cup K}$ such pairs. For each fixed $(a,b)$,
	there are $\binom{|D_I(a,b)|}{\ell}$ such subsets $K$ (zero if
	$|D_I(a,b)|<\ell$). Therefore
	\begin{equation}\label{eq:counting}
		\begin{aligned}
			\sum_{\substack{K\subseteq I^c\\|K|=\ell}}N_{I\cup K}
			&=\sum_{(a,b)\in\cE_I}\binom{|D_I(a,b)|}{\ell} =N_I\E\binom W\ell.
		\end{aligned}
	\end{equation}
	The last equality uses that $(X,Y)$ is uniform on the $N_I$ pairs
	in $\cE_I$ and that $W=|D_I(X,Y)|$.
	On the nonnegative integers, the sequence $w\mapsto\binom w\ell$
	uses the convention $\binom w\ell=0$ for $w<\ell$.
	Its first differences are $\binom w{\ell-1}$, which are nonnegative
	and nondecreasing. Its piecewise-linear interpolation is therefore
	nondecreasing and convex. Jensen's inequality and $\E W\ge qm$ give
	\[
	\E\binom W\ell\ge\binom{\lfloor qm\rfloor}{\ell}.
	\]
	Averaging \eqref{eq:counting} over all $\binom m\ell$ choices of $K$
	proves the result.
\end{proof}
%

The following initial extraction is the combination of Zhang's
orientation lemma and section-counting lemma
\cite[Lemmas 3.1 and 3.5]{Zhang}. His set $\Delta_I$ is our $\cE_I$;
complementing the first or the second codebook is equivalent by user
symmetry. We include the proof for completeness.

\begin{lemma}[Initial orientation and projection \cite{Zhang}]\label{lem:initial}
For every zero-error code $A,B$ and every integer
$1\le k\le\lfloor n/2\rfloor$, either the original code or the code
$A,\one-B$ has a section $I$ of size $k$ satisfying
\begin{equation}\label{eq:initial-section}
N_I\ge |A||B|\frac{\binom{\lfloor n/2\rfloor}{k}}{\binom nk}.
\end{equation}
\end{lemma}

\begin{proof}
Replacing $B$ by $\one-B$ preserves the zero-error property. 
For independent uniform inputs on
$A$ and $B$, this replacement changes the mean Hamming distance from
$\E\dist(X,Y)$ to $n-\E\dist(X,Y)$. We may therefore orient the code so
that the mean distance is at least $n/2$. Apply
Lemma~\ref{lem:projection} to $I=\varnothing$, $q=1/2$, and $\ell=k$.
\end{proof}

\section{Finite-blocklength and asymptotic converses}\label{sec:finite}

\begin{theorem}[Finite-blocklength iterated converse]\label{thm:finite}
Fix integers
\[
1\le k_0<k_1<\cdots<k_K<n,\qquad k_0\le\lfloor n/2\rfloor.
\]
Write $m_j=n-k_j$ and $\ell_j=k_{j+1}-k_j$ for $0\le j<K$.
Choose $q_j\in(0,1/2)$ such that
$\ell_j\le\lfloor q_jm_j\rfloor$.
Starting with $B_K=m_K\log3$, define backwards
\begin{equation}\label{eq:finite-recursion}
B_j=\max\left\{
 k_j+m_jJ(q_j),\quad
 B_{j+1}+\log\frac{\binom{m_j}{\ell_j}}
 {\binom{\lfloor q_jm_j\rfloor}{\ell_j}}
\right\}.
\end{equation}
Then
\begin{equation}\label{eq:finite-bound}
\log M(n)\le
\log\frac{\binom n{k_0}}{\binom{\lfloor n/2\rfloor}{k_0}}+B_0.
\end{equation}
\end{theorem}

\begin{proof}
We first show, by backward induction on $j$, that every nonempty
section $\cE_I$ with $|I|=k_j$ satisfies $\log N_I\le B_j$.
At stage $K$, the remaining ternary output uniquely determines the pair,
so $N_I\le3^{m_K}$.

For the induction step, if $q_I\le q_j$,
Lemma~\ref{lem:section-entropy} gives
$\log N_I\le k_j+m_jJ(q_j)$.
If $q_I>q_j$, Lemma~\ref{lem:projection} produces a nonempty section
$\cE_{I\cup K'}$ with $|I\cup K'|=k_{j+1}$ with
\[
\log N_I\le
\log N_{I\cup K'}+
\log\frac{\binom{m_j}{\ell_j}}
{\binom{\lfloor q_jm_j\rfloor}{\ell_j}}
\le B_{j+1}+
\log\frac{\binom{m_j}{\ell_j}}
{\binom{\lfloor q_jm_j\rfloor}{\ell_j}}.
\]
Both cases are bounded by \eqref{eq:finite-recursion}.
Finally, orient the original code and extract an initial section as in
Lemma~\ref{lem:initial}. Applying $\log N_I\le B_0$ to this section
proves \eqref{eq:finite-bound}.
\end{proof}

For $0<a<q<1$, define
\begin{equation}\label{eq:d}
d(a,q):=h(a)-q\,h(a/q).
\end{equation}

\begin{theorem}[Asymptotic iterated converse]\label{thm:asymptotic}
Fix an integer $K\ge1$ and real numbers
\[
0<s_0<s_1<\cdots<s_K<1,\qquad s_0<\frac12.
\]
For $0\le j<K$, choose $q_j$ satisfying
\begin{equation}\label{eq:alpha-domain}
\alpha_j:=\frac{s_{j+1}-s_j}{1-s_j}<q_j<\frac12.
\end{equation}
Set $T_K=(1-s_K)\log3$ and define backwards
\begin{equation}\label{eq:Trec}
T_j=\max\left\{
 s_j+(1-s_j)J(q_j),\quad
 T_{j+1}+(1-s_j)d(\alpha_j,q_j)
\right\}.
\end{equation}
Then
\begin{equation}\label{eq:asymptotic-bound}
\Cs\le h(s_0)-\frac12h(2s_0)+T_0.
\end{equation}
\end{theorem}

\begin{proof}
In Theorem~\ref{thm:finite}, take $k_j=\lfloor s_jn\rfloor$.
For all sufficiently large $n$, all integer-domain conditions hold
because the inequalities in \eqref{eq:alpha-domain} are strict.
The elementary binomial asymptotic
\[
\log\binom Nr=Nh(r/N)+O(\log(N+1))
\]
gives
\begin{align*}
\frac1n\log\frac{\binom n{k_0}}
{\binom{\lfloor n/2\rfloor}{k_0}}
&\longrightarrow h(s_0)-\frac12h(2s_0),\\
\frac1n\log\frac{\binom{m_j}{\ell_j}}
{\binom{\lfloor q_jm_j\rfloor}{\ell_j}}
&\longrightarrow (1-s_j)d(\alpha_j,q_j).
\end{align*}
Backward induction in \eqref{eq:finite-recursion}, using continuity of
maximum, shows that $B_j/n\to T_j$. Divide \eqref{eq:finite-bound}
by $n$ and pass to the limit. The number of stages $K$ is fixed during
this blocklength limit. 
\end{proof}

\section{The continuous projection converse}\label{sec:continuous}

The finite-stage recursion has a useful continuous limit. Its proof uses
Theorem~\ref{thm:asymptotic} with each number of stages fixed first, and
therefore does not require estimates uniform in both blocklength and the
number of stages. We obtain Theorem \ref{thm:continuous}, which is repeated here. 

\begin{theorem}[Continuous projection converse]\label{thm:continuous2}
Let $0<a<b<1$ with $a<1/2$, and let
$q\in C^2([a,b])$ take values in $(0,1/2)$. Define
\[
U(s)=s+(1-s)J(q(s)).
\]
Suppose
\begin{equation}\label{eq:potential-conditions}
U(b)\ge(1-b)\log3,\qquad U'(s)\le\log q(s)
\quad(a\le s\le b).
\end{equation}
Then
\begin{equation}\label{eq:continuous-bound}
\Cs\le h(a)-\frac12h(2a)+a+(1-a)J(q(a)).
\end{equation}
\end{theorem}

\begin{proof}
For an integer $K$, put $\varepsilon=(b-a)/K$,
$s_j=a+j\varepsilon$, and $q_j=q(s_j)$ for $0\le j<K$.
Since $q$ is bounded away from zero and $b<1$, the conditions
$\alpha_j=\varepsilon/(1-s_j)<q_j<1/2$ hold for all sufficiently
large $K$.

Uniformly on the compact range of these parameters,
\begin{equation}\label{eq:small-projection}
d(\alpha,q)=\alpha\log(1/q)+O(\alpha^2).
\end{equation}
For example, this follows by integrating
\[
\frac{\partial}{\partial\alpha}d(\alpha,q)
=\log\frac{1-\alpha}{q-\alpha},
\]
whose value at $\alpha=0$ is $\log(1/q)$ and whose derivative is
bounded on the relevant compact set. Taylor's theorem and
\eqref{eq:potential-conditions} consequently give a constant $C>0$, independent
of $K$ and $j$, such that
\begin{align*}
U(s_{j+1})+(1-s_j)d(\alpha_j,q_j)
&\le U(s_j)+\varepsilon\bigl(U'(s_j)-\log q(s_j)\bigr)
  +C\varepsilon^2\\
&\le U(s_j)+C\varepsilon^2.
\end{align*}
The terminal condition and backward induction in \eqref{eq:Trec} imply
\[
T_j\le U(s_j)+C(K-j)\varepsilon^2.
\]
Indeed, the first branch of that recursion is exactly $U(s_j)$, and
the second branch obeys the preceding estimate.
Theorem~\ref{thm:asymptotic} therefore yields
\[
\Cs\le h(a)-\tfrac12h(2a)+U(a)+\frac{C(b-a)^2}{K}.
\]
This holds for every sufficiently large fixed $K$, after taking the
blocklength limit. Letting $K\to\infty$ proves the claim.
\end{proof}

\section{An explicit trajectory and a certified upper bound}\label{sec:explicit}

\subsection{Intuition: balancing stopping and further projection}

The trajectory balances the two alternatives in the recursion at each
stage. Suppose a fraction $s$ of the coordinates has been fixed, and
choose a disagreement threshold $q(s)$. If the remaining disagreement
is small, the conditional entropy estimate gives the normalized bound
\[
U(s):=s+(1-s)J(q(s)).
\]
If the disagreement is large, we can project further, from $s$ to
$s+\delta$. By \eqref{eq:small-projection}, the normalized logarithmic
cost of this extraction is
\[
(1-s)d\!\left(\frac{\delta}{1-s},q(s)\right)
=\delta\log\frac1{q(s)}+O(\delta^2).
\]
Thus, if the later section admits the bound $U(s+\delta)$, continuing
with another projection gives
\[
U(s+\delta)+\delta\log\frac1{q(s)}+O(\delta^2).
\]
To make stopping with the entropy estimate and continuing with the
projection equally strong to first order, we impose
\[
U(s)=U(s+\delta)+\delta\log\frac1{q(s)}+O(\delta^2).
\]
Taylor expansion yields the balance equation
\[
U'(s)=\log q(s).
\]
This explains the equality case of the differential inequality in
Theorem~\ref{thm:continuous2}. Substituting
$U(s)=s+(1-s)J(q(s))$ gives
\[
q'(s)=-\frac{1-J(q(s))-\log q(s)}{(1-s)J'(q(s))}.
\]
On the interval used below, the numerator and denominator are positive,
so the threshold decreases as more coordinates are fixed.

At the terminal point, we make the entropy bound meet the elementary
output-counting bound:
\[
b+(1-b)J(q(b))=(1-b)\log3.
\]
We choose $q(b)=1/3$. This is the disagreement probability associated
with a uniform ternary output: disagreement is exactly the event that
the output equals one. It is also the maximizer of
$L(q)=h(q)+1-q$, the entropy bound that uses only the disagreement
probability. The terminal equality then gives
\[
\frac{b}{1-b}=\log3-J(1/3).
\]
This motivates the endpoint without asserting that the actual section
has a uniform output distribution.

Finally, we follow the balanced trajectory backward and choose a
starting point. Starting later increases the initial extraction cost
$h(s)-\tfrac12h(2s)$, but decreases the remaining-section bound $U(s)$.
Balancing these effects means minimizing
\[
\Phi(s):=h(s)-\tfrac12h(2s)+U(s).
\]
As derived below, an interior stationary point satisfies
$q(s)=(1/2-s)/(1-s)$. These observations motivate the construction;
they do not establish global optimality among all trajectories or
converse methods. The rigorous bound follows from
Theorem~\ref{thm:continuous} and the explicit verification below.

\subsection{Construction of the balanced trajectory}

We now construct a trajectory satisfying both conditions in
\eqref{eq:potential-conditions} with equality. Write
\begin{align}
	G(v)&:=1-J(v)-\log v,\label{eq:G}\\
	c&:=\log3-J(1/3),\qquad b:=\frac{c}{1+c},\label{eq:terminal}\\
	\Psi(r)&:=\int_{1/3}^{r}\frac{J'(v)}{G(v)}\,dv,
	\qquad
	\sigma(r):=1-\frac{\exp(\Psi(r))}{1+c}.
	\label{eq:trajectory}
\end{align}
The exponential in \eqref{eq:trajectory} is the natural exponential;
the logarithm in $G$ remains base two. On $[1/3,1/2)$,
$J'>0$ and $G\ge1/2$, since $J\le3/2$ and $\log v\le-1$.
Also $c>0$: at disagreement $1/3$, the independent-bit entropy maximizer
from Lemma~\ref{lem:bits} is not the uniform ternary distribution.
Thus $\sigma$ is smooth and strictly decreasing, and
$\sigma(1/3)=b$.

Choose the single rational parameter
\begin{equation}\label{eq:rational-parameter}
	r_*:=\frac{23750713}{50000000}=0.47501426,
	\qquad a:=\sigma(r_*).
\end{equation}
The certified enclosures below imply $0<a<b<1/2$. Define $q(s)$ for
$a\le s\le b$ to be the inverse of $\sigma$ restricted to
$[1/3,r_*]$. In particular, $q(a)=r_*$ and $q(b)=1/3$.

\begin{proposition}[Explicit continuous bound]\label{prop:explicit}
For $a$ specified by \eqref{eq:trajectory}--\eqref{eq:rational-parameter},
\begin{equation}\label{eq:explicit-bound}
\Cs\le \mathcal B
:=h(a)-\frac12h(2a)+a+(1-a)J(r_*).
\end{equation}
\end{proposition}

\begin{proof}
Differentiating \eqref{eq:trajectory} gives
\[
\sigma'(r)=-(1-\sigma(r))\frac{J'(r)}{G(r)},\qquad
q'(s)=-\frac{G(q(s))}{(1-s)J'(q(s))}.
\]
For $U(s)=s+(1-s)J(q(s))$ it follows that
\begin{align}
U'(s)
&=1-J(q(s))+(1-s)J'(q(s))q'(s)\\
&=1-J(q(s))-G(q(s))=\log q(s). \label{eq:U-derivative}
\end{align}
Furthermore, $b=(1-b)c$, so
\[
U(b)=b+(1-b)J(1/3)=(1-b)\log3.
\]
Theorem~\ref{thm:continuous2} applies.
\end{proof}


All decimal endpoints below are exact rational numbers. We use  
interval arithmetic to establish
\begin{align}
0.30083655751150239577
&<b<0.30083655751150239578,\label{eq:b-enclosure}\\
0.3091077386432610
&<\Psi(r_*)<0.3091077386432632,\label{eq:psi-enclosure}\\
0.0475931786634061
&<a<0.0475931786634082,\label{eq:a-enclosure}\\
1.4800634255387845
&<\mathcal B<1.4800634255388169.\label{eq:B-enclosure}
\end{align}
In particular,
\begin{equation}\label{eq:main-result}
\Cs<1.480063425539<1.480063426.
\end{equation}

\begin{remark}[Choice of the rational parameter]
The proof checks the explicit rational $r_*$ and does not assume an
optimization succeeds. To motivate this choice, differentiate the bound
along the constructed trajectory (using \eqref{eq:U-derivative}):
\[
\frac{d}{ds}\left(h(s)-\tfrac12h(2s)+U(s)\right)
=\log\frac{q(s)(1-s)}{1/2-s}.
\]
Thus a stationary starting point satisfies
$q(s)=(1/2-s)/(1-s)$. The rational parameter in
\eqref{eq:rational-parameter} is close to such a point. No assertion
of global optimality among all possible converses is needed or made.
\end{remark}

\section{Achievability by weighted complement gluing}\label{sec:achievability}

The complement-gluing construction of Balogh, Nguyen, \"Osterg\aa rd,
White, and Wigal~\cite{BNOWW} gives a two-user sum rate of
$1.318446971988\ldots$ using ordinary Hamming weight. We replace that
statistic by a positive integer weighted sum of coordinates. Its additivity
still separates the two output families, while the different coordinate
weights give a larger retained message product.

\subsection{A weighted gluing lemma}

Let $A_0,B_0\subseteq\{0,1\}^d$ be a zero-error code pair, and let
$\mathbf{w}=(w_0,\ldots,w_{d-1})\in\mathbb Z_{\ge0}^d$ have positive sum $W$.
For $\mathbf{x}=(x^{(1)},\ldots,x^{(m)})\in\{0,1\}^{dm}$, define
\[
S_\mathbf{w}(\mathbf{x}):=\sum_{r=1}^m\sum_{j=0}^{d-1}w_j x^{(r)}_j,
\qquad D:=mW.
\]
Use the same linear formula on ternary outputs. Then
$S_\mathbf{w}(\mathbf{x}+\mathbf{y})=S_\mathbf{w}(\mathbf{x})+S_\mathbf{w}(\mathbf{y})$ and
$S_\mathbf{w}(\one-\mathbf{x})=D-S_\mathbf{w}(\mathbf{x})$.

\begin{lemma}[Weighted complement gluing]\label{lem:weighted-gluing}
Fix an integer $\ell$ with $0<2\ell\le D$, and put
\begin{align}
A_-&:=\{\mathbf{x}\in A_0^m:S_\mathbf{w}(\mathbf{x})\le\ell-1\},\nonumber\\
A_+&:=\{\one-\mathbf{x}:\mathbf{x}\in A_0^m,\ S_\mathbf{w}(\mathbf{x})\le\ell\},\nonumber\\
\widetilde A&:=A_-\cup A_+,\qquad
\widetilde B:=\{\mathbf{y}\in B_0^m:\ell\le S_\mathbf{w}(\mathbf{y})\le D-\ell\}.
\label{eq:weighted-code}
\end{align}
If both resulting codebooks are nonempty, $(\widetilde A,\widetilde B)$
is a length-$dm$ zero-error code pair, and the union defining
$\widetilde A$ is disjoint.
\end{lemma}

\begin{proof}
Tensor products preserve unique decodability. Complementing the entire
first constituent also preserves it, since
\[
((\one-A_0^m)-(\one-A_0^m))
=-(A_0^m-A_0^m)=A_0^m-A_0^m.
\]
Thus $(A_-,\widetilde B)$ and $(A_+,\widetilde B)$ are individually
zero-error code pairs. For their outputs, respectively, one has
\begin{align*}
S_\mathbf{w}(\mathbf{a}_-+\mathbf{b})&\le(\ell-1)+(D-\ell)=D-1,\\
S_\mathbf{w}(\mathbf{a}_++\mathbf{b})&\ge(D-\ell)+\ell=D.
\end{align*}
Consequently the two output sets are disjoint, proving unique decodability
of their union. The two input families are also disjoint because
$S_\mathbf{w}(\mathbf{a}_-)\le\ell-1<D-\ell\le S_\mathbf{w}(\mathbf{a}_+)$.
\end{proof}

The receiver first compares the weighted output score with $D$ to identify
the family, and then uses the unique decoder for that family. Each encoder
chooses its word from its own fixed codebook, so this construction uses
neither feedback nor a shared message.

\subsection{An explicit length-672 construction}

Use the six-coordinate seed listed in \cite[Section 4.1]{BNOWW}:
\begin{align}
A_0&=\{3,4,7,10,14,17,21,27,32,36,42,49,56,59,60\},\nonumber\\
B_0&=\{8,9,16,18,24,29,30,31,32,33,34,39,45,47,54,55\}.
\label{eq:achievable-seed}
\end{align}
An integer label $x$ represents the vector $(x_0,\ldots,x_5)$ determined
by $x=\sum_{j=0}^5 2^j x_j$; thus coordinate zero is the least significant
bit. Integer labels are only notation: channel addition is coordinatewise,
without carries. Direct enumeration verifies that the $15\cdot16=240$
ternary sums are all distinct.

Choose
\begin{equation}\label{eq:achievable-parameters}
 w=(2,2,9,3,3,4),\qquad m=112,\qquad \ell=1188.
\end{equation}
Then $W=23$, $D=2576$, and $dm=672$. Thus $\widetilde A$ consists of
the words in $A_0^{112}$ of score at most 1187, together with the
complements of the words of score at most 1188; $\widetilde B$ consists
of the words in $B_0^{112}$ with scores between 1188 and 1388 inclusive.
The two output families have scores at most 2575 and at least 2576,
respectively.

To count these codebooks exactly, define the seed score polynomials
\begin{align}
P_A(z)&:=\sum_{\mathbf{x}\in A_0}z^{S_\mathbf{w}(\mathbf{x})}
=2z^4+2z^5+3z^9+2z^{10}+2z^{13}+3z^{14}+z^{19},\nonumber\\
P_B(z)&:=\sum_{\mathbf{y}\in B_0}z^{S_\mathbf{w}(\mathbf{y})}
=2z^3+z^4+2z^5+3z^6+3z^{17}+2z^{18}+z^{19}+2z^{20}.
\label{eq:score-polynomials}
\end{align}
Writing $a_k=[z^k]P_A(z)^{112}$ (the coefficient of $z^k$ in $P_A(z)^{112}$) and $b_k=[z^k]P_B(z)^{112}$ gives
\begin{equation}\label{eq:achievable-counts}
N_A:=|\widetilde A|=2\sum_{k=0}^{1187}a_k+a_{1188},\qquad
N_B:=|\widetilde B|=\sum_{k=1188}^{1388}b_k.
\end{equation}

\begin{theorem}[Certified achievable rate]\label{thm:achievable}
The construction \eqref{eq:weighted-code} with
\eqref{eq:achievable-seed}--\eqref{eq:achievable-parameters} satisfies
\begin{equation}\label{eq:achievable-bound}
 \Cs\ge R_*:=\frac{\log(N_A N_B)}{672}>1.318639029203.
\end{equation}
More precisely, the exact code rate is enclosed by
\begin{equation}\label{eq:achievable-enclosure}
1.31863902920327531492028443956481
<R_*<
1.31863902920327531492028443956483.
\end{equation}
\end{theorem}

\begin{proof}
The seed check and Lemma~\ref{lem:weighted-gluing} establish unique
decodability, and \eqref{eq:achievable-counts} gives the exact message
counts. Concatenating this fixed code with itself makes $R_*$ achievable
asymptotically; equivalently use \eqref{eq:capacity}. Exact integer
coefficient computation followed by outward-rounded logarithm evaluation
gives \eqref{eq:achievable-enclosure}.
\end{proof}

\subsection{Comparison with unweighted gluing}

For comparison, the unweighted
parameters $w=(1,1,1,1,1,1)$, $m=142$, $\ell=402$, give the construction in \cite{BNOWW}, whose rate satisfies
\[
1.31844697198813872162<R_{\rm unweighted}
<1.31844697198813872164.
\]
The improvement is approximately $0.0001920572151$ bits per channel use,
with strict improvement certified by disjoint rate intervals.

The seed's coordinate-one counts are $(7,7,6,7,7,7)$ for $A_0$ and
$(8,8,8,8,8,8)$ for $B_0$. This motivates using unequal weights to exploit
the different coordinate imbalances. The specific weights and cutoff
were chosen numerically; the proof requires only the fixed values in
\eqref{eq:achievable-parameters}. No optimality assertion about the seed,
weights, or gluing construction is made.

\section*{Acknowledgements }

\emph{Generative-AI use disclosure:}  ChatGPT 5.6 Sol assisted with mathematical arguments and code constructions,
 proof drafting, and manuscript
 preparation. 
 The author independently verified all arguments and mathematical claims.  The author retain responsibility
 for the correctness, originality, and attribution of all content.
 
\emph{Funding:} This work was supported by the National Key Research
and Development Program of China under grant 2023YFA1009604 and the
NSFC under grant 62101286.

\bibliographystyle{unsrt}
\bibliography{ref}

\begin{thebibliography}{10}

\bibitem{OS}
Ofer Ordentlich and Ofer Shayevitz.
\newblock An upper bound on the sizes of multiset-union-free families.
\newblock {\em SIAM Journal on Discrete Mathematics}, 30(2):1032--1045, 2016.

\bibitem{Lindstrom}
Bengt Lindstr{\"o}m.
\newblock Determination of two vectors from the sum.
\newblock {\em Journal of Combinatorial Theory}, 6(4):402--407, 1969.

\bibitem{MO}
Michael Mattas and Patric R.~J. {\"O}sterg{\aa}rd.
\newblock A new bound for the zero-error capacity region of the two-user binary
  adder channel.
\newblock {\em IEEE Transactions on Information Theory}, 51(9):3289--3291,
  2005.

\bibitem{BNOWW}
J{\'a}nos Balogh, Tuan Nguyen, Patric R.~J. {\"O}sterg{\aa}rd, Edward~P. White,
  and Michael~C. Wigal.
\newblock Improving uniquely decodable codes in binary adder channels.
\newblock {\em IEEE Transactions on Information Theory}, 71(5):3027--3038, May
  2025.

\bibitem{AKKN}
Per Austrin, Petteri Kaski, Mikko Koivisto, and Jesper Nederlof.
\newblock Sharper upper bounds for unbalanced uniquely decodable code pairs.
\newblock {\em IEEE Transactions on Information Theory}, 64(2):1368--1373,
  2018.

\bibitem{OPS}
Ofer Ordentlich, Yury Polyanskiy, and Or~Shayevitz.
\newblock A note on the probability of rectangles for correlated binary
  strings.
\newblock {\em IEEE Transactions on Information Theory}, 66(12):7878--7886,
  December 2020.

\bibitem{YAC}
Lei Yu, Venkat Anantharam, and Jun Chen.
\newblock Graphs of joint types, noninteractive simulation, and stronger
  hypercontractivity.
\newblock {\em IEEE Transactions on Information Theory}, 70(4):2287--2308,
  2024.

\bibitem{Zhang}
Ru~Zhang.
\newblock Improved upper bound for lindstr{\"o}m's unique-sum problem.
\newblock {\em arXiv preprint arXiv:2608.05762}, 2026.

\bibitem{KLWY}
Toshio Kasami, Shu Lin, Victor~K. Wei, and Shigeru Yamamura.
\newblock Graph theoretic approaches to the code construction for the two-user
  multiple-access binary adder channel.
\newblock {\em IEEE Transactions on Information Theory}, 29(1):114--130, 1983.

\bibitem{Wiman}
Mikael Wiman.
\newblock Improved constructions of unbalanced uniquely decodable code pairs.
\newblock Master's thesis, KTH Royal Institute of Technology, 2017.

\bibitem{Wyner}
Aaron~D. Wyner.
\newblock The common information of two dependent random variables.
\newblock {\em IEEE Transactions on Information Theory}, 21(2):163--179, March
  1975.

\end{thebibliography}

\end{document}